\documentclass[a4paper, USenglish, cleveref, autoref, thm-restate]{lipics-v2021}

\hideLIPIcs
\nolinenumbers

\title{Toward Better Depth Lower Bounds: Strong Composition of XOR and a Random Function}
\titlerunning{Strong Composition of XOR and a Random Function}

\author{Nikolai Chukhin}{Neapolis University Pafos \& JetBrains Research}{buyolitsez1951@gmail.com}{}{}
\author{Alexander~S. Kulikov}{JetBrains Research \and\url{https://alexanderskulikov.github.io}}{alexander.s.kulikov@gmail.com}{https://orcid.org/0000-0002-5656-0336}{}
\author{Ivan Mihajlin}{JetBrains Research}{ivmihajlin@gmail.com}{}{}
\authorrunning{N.~Chukhin, A.~Kulikov, I.~Mikhajlin}

\Copyright{Nikolai Chukhin, Alexander~S. Kulikov, Ivan Mikhajlin}


\EventEditors{Olaf Beyersdorff, Micha\l{} Pilipczuk, Elaine Pimentel, and Nguyen Kim Thang}
\EventNoEds{4}
\EventLongTitle{42nd International Symposium on Theoretical Aspects of Computer Science (STACS 2025)}
\EventShortTitle{STACS 2025}
\EventAcronym{STACS}
\EventYear{2025}
\EventDate{March 4--7, 2025}
\EventLocation{Jena, Germany}
\EventLogo{}
\SeriesVolume{327}
\ArticleNo{22}

\DeclareMathOperator{\avgdeg}{avgdeg}
\DeclareMathOperator{\proj}{proj}

\usepackage{complexity}
\newfunc{\CC}{CC}
\newfunc{\KW}{KW}
\newlang{\XOR}{XOR}
\newlang{\OR}{OR}

\begin{document}

\sloppy

\maketitle

\begin{abstract}
	Proving formula depth lower bounds is a~fundamental challenge in~complexity theory, with the strongest known bound of \((3 - o(1))\log n\) established by H{\aa}stad over 25 years ago. The Karchmer--Raz--Wigderson (KRW) conjecture offers a~promising approach to~advance these bounds and separate
    $\P$~from~$\NC^{1}$. It~suggests that the depth complexity of a~function composition
    $f \diamond g$ approximates the sum of the depth complexities of $f$~and~$g$.

	The Karchmer--Wigderson (KW) relation framework translates formula depth into communication complexity, restating the KRW conjecture as \(\mathsf{CC}(\mathsf{KW}_f \diamond \mathsf{KW}_g) \approx \mathsf{CC}(\mathsf{KW}_f) + \mathsf{CC}(\mathsf{KW}_g)\). Prior work has confirmed the conjecture under various relaxations, often replacing one or both KW relations with the universal relation or constraining the communication game through strong composition.

    In this paper, we examine the strong composition \(\mathsf{KW}_{\mathsf{XOR}} \circledast \mathsf{KW}_f\) of the parity function and a~random Boolean function~$f$.
	We prove that with probability $1-o(1)$, any protocol solving this composition requires at least \(n^{3 - o(1)}\) leaves. This result establishes a~depth lower bound of \((3 - o(1))\log n\), matching H{\aa}stad's bound, but is~applicable to a broader class of inner functions, even when the outer function is simple.
    Though bounds for the strong composition do not translate directly to
    formula depth bounds, they usually help to analyze the standard composition
    (of the corresponding two functions) which is directly related to formula depth.

    Our proof utilizes formal complexity measures. First, we apply Khrapchenko's method to show that numerous instances of~\(f\) remain unsolved after several communication steps. Subsequently, we transition to a~different formal complexity measure to demonstrate that the remaining communication problem is at least as hard as \(\mathsf{KW}_{\mathsf{OR}} \circledast \mathsf{KW}_f\). This hybrid approach not only achieves the desired lower bound, but also introduces a~novel technique for analyzing formula depth, potentially informing future research in complexity theory.
\end{abstract}

\section{Introduction}
Proving formula depth lower bounds is~an~important and difficult challenge
in~complexity theory: the strongest known lower bound $(3-o(1))\log n$
proved by~H{\aa}stad~\cite{DBLP:journals/siamcomp/Hastad98} (following a line of works starting from Subbotovskaya \cite{Sub61,DBLP:journals/rsa/ImpagliazzoN93,DBLP:journals/rsa/PatersonZ93}) remains unbeaten
for more than 25~years already (in~2014, Tal~\cite{DBLP:conf/focs/Tal14}
improved lower order terms in~this lower bound).
One of~the most actively studied
approaches to~this problem is~the one suggested by~Karchmer, Raz, and Wigderson~\cite{DBLP:journals/cc/KarchmerRW95}.
They conjectured that the naive approach of~computing a~composition
of~two functions is~close to~optimal. Namely, for two Boolean
functions $f \colon \{0,1\}^m \to \{0,1\}$ and $g \colon \{0,1\}^n \to \{0,1\}$,
define their composition $f \diamond g \colon \{0,1\}^{m \times n} \to \{0,1\}$
as~a~function that first applies~$g$ to~every row of~the input matrix
and then applies~$f$ to~the resulting column vector. The KRW conjecture then states
that $\D(f \diamond g)$ is~close to~$\D(f)+\D(g)$, where $\D(\cdot)$~denotes the minimum depth
of a~de~Morgan formula computing the given function.
Karchmer, Raz, and Wigderson~\cite{DBLP:journals/cc/KarchmerRW95} proved that
if~the conjecture is~true, then $\P \not \subseteq \NC^1$, that~is, there are functions
in~$\P{}$ that cannot be~computed in~logarithmic parallel time.

A~convenient way of~studying the KRW conjecture is~through the framework
of~Karchmer--Wigderson relation~\cite{DBLP:journals/siamdm/KarchmerW90}.
It~not only allows one to~apply the tools
from communication complexity, but also suggests various important special
cases of~the conjecture. For a~function $f \colon \{0,1\}^n \to \{0,1\}$,
the relation $\KW_f$ is~defined as~follows:
\[\KW_f=\{(a, b, i) \colon a \in f^{-1}(1), b \in f^{-1}(0), i \in [n], a_i \neq b_i\}.\]
The communication complexity $\CC(\KW_f)$ of~this relation is~the minimum number of~bits that Alice and Bob need to~exchange to~solve the following communication problem:
Alice is~given $a \in f^{-1}(1)$,
Bob is~given $b \in f^{-1}(0)$, and their goal is~to~find an~index $i \in [n]$ such that $(a,b,i) \in \KW_f$ (i.e., $a_i \neq b_i$).
Karchmer and Wigderson~\cite{DBLP:journals/siamdm/KarchmerW90}
proved that, for any function~$f$, the communication complexity of~$\KW_f$
is~equal to~the depth complexity of~$f$: $\CC(\KW_f)=\D(f)$. Within this framework,
the KRW conjecture is~restated as~follows: $\CC(\KW_f \diamond \KW_g)$
is~close to~$\CC(\KW_f)+\CC(\KW_g)$ (where $\KW_f \diamond \KW_g$
is~another name for $\KW_{f \diamond g}$).

One natural way of~relaxing the conjecture is~to~replace one or~both of~the two relations $\KW_f$ and~$\KW_g$ by the universal relation, defined as~follows:
\[U_n=\{(a, b, i) \colon a, b \in \{0,1\}^n, a \neq b, i \in [n], a_i \neq b_i\}.\]
Using a~universal relation instead of~the Karchmer--Wigderson relation makes the corresponding communication game only harder, hence proving lower bounds
for~it is~potentially easier and could lead to~the resolution of~the original conjecture. For this reason, such relaxations have been studied intensively.

Edmonds et al. \cite{DBLP:journals/cc/EdmondsIRS01} proved the KRW conjecture for the composition $U_m \diamond U_n$ of~two universal relations using communication complexity methods.
H{\aa}stad and Wigderson~\cite{DBLP:conf/dimacs/HastadW90} improved~it for a~higher degree of~composition using a~different approach.
Karchmer et~al.~\cite{DBLP:journals/cc/KarchmerRW95} extended this result
to~monotone functions.
H{\aa}stad~\cite{DBLP:journals/siamcomp/Hastad98} demonstrated the conjecture for the composition $f \diamond \XOR_n$ of an~arbitrary function $f \colon \{0,1\}^m \to \{0,1\}$ with the parity function $\XOR_n$.
This was later reaffirmed by~Dinur and Meir~\cite{DBLP:journals/cc/DinurM18} through a~communication complexity approach.
Further advancements were made
by~Gavinsky et~al.~\cite{DBLP:journals/siamcomp/GavinskyMWW17} who established the conjecture for the composition $f \diamond U_n$ of~any non-constant function $f \colon \{0,1\}^m \to \{0,1\}$ with the~universal relation~$U_n$.
Mihajlin and Smal~\cite{DBLP:conf/coco/MihajlinS21} proved the KRW conjecture for the composition of~a~universal relation with certain hard functions using $\XOR$-composition.
Subsequently, Wu~\cite{DBLP:journals/corr/abs-2310-07422} improved this result by~extending~it to~the composition of a~universal relation with a~wider range of functions (though still not with the majority of~them).
de~Rezende et~al.~\cite{DBLP:journals/cc/RezendeMNPR24} proved the conjecture
in~a~semi-monotone setting for a~wide range of~functions~$g$.

Another natural way of~relaxing the initial conjecture is~to~constrain the
communication game (instead of~allowing for more inputs for the game).
In~the \emph{strong composition} $\KW_f \circledast \KW_g$, Alice receives $X \in (f \diamond g)^{-1}(1)$ and Bob receives $Y \in (f \diamond g)^{-1}(0)$, and
their objective is~to~identify a~pair of~indices $(i, j)$ such that $X_{i, j} \neq Y_{i, j}$, similar to~the regular composition.
However, this~time it~must hold additionally
that $g(X_i) \neq g(Y_i)$.

This way of~relaxing the conjecture was considered in a~number of~previous papers and was
formalized recently by~Meir~\cite{DBLP:conf/focs/Meir23}.
H{\aa}stad and Wigderson, in~their proof of~the lower bound for two universal relations, initially establish the result for what they call the extended universal relation, a~concept closely related to strong composition.
Similarly, Karchmer et al. \cite{DBLP:journals/cc/KarchmerRW95} demonstrate that,
in~the monotone setting, strong composition coincides with the standard composition.
de Rezende et~al.~\cite{DBLP:journals/cc/RezendeMNPR24} utilized this notion, although without explicitly naming~it. Meir~\cite{DBLP:conf/focs/Meir23} formalized the notion
of~strong composition in~his proof of~the
relaxation of~the KRW conjecture.
\begin{theorem}[Meir, \cite{DBLP:conf/focs/Meir23}]
    \label{thm:meir}
    There exists a constant $\gamma > 0.04$ such that for every non-constant function $f \colon \{0,1\}^{m} \to \{0,1\}$ and for all $n \in \mathbb N$, there exists a function $g \colon \{0,1\}^{n} \to \{0,1\}$ such that
    \[
    \CC(\KW_f \circledast \KW_g) \geq \CC(\KW_f) - (1 - \gamma) m + n - O(\log (m n)).
    \]
\end{theorem}

\subsection{Our Result}
H{\aa}stad~\cite{DBLP:journals/siamcomp/Hastad98} proved the KRW conjecture for $\KW_f \diamond \KW_{\XOR}$.
However it is still an open question to prove the KRW conjecture for $\KW_{\XOR} \diamond \KW_f$.
In~this paper, we~study the strong composition $\KW_{\XOR_m} \circledast \KW_f$
of~the parity function $\XOR_m$
with a~random function $f \colon \{0,1\}^{\log m} \to \{0,1\}$.
Since Alice and~Bob receive an~input of~size $m \log m$, we~estimate the~size of $\KW_{\XOR_m} \circledast \KW_f$ in~terms of~$n = m \log m$.
It~is not difficult to~see that the communication complexity
of~the corresponding game is~at~most $3\log n$:
$\KW_f$ can be~solved in~$\log m$ bits of~communication,
whereas $\KW_{\XOR_m}$ can be~solved in $2\log m$ bits of~communication, using the standard divide-and-conquer approach (Alice sends the parity of~the first half,
Bob then identifies the half in~which the parity differs,
thus, by utilizing 2~bits of communication, the input size is reduced by a factor of two).
We~prove that
if~the function~$f$ is~well balanced and hard to~approximate (which happens with probability $1 - o(1)$),
then the bound $3\log n$ is~essentially optimal.
Below, we~state the result
in~terms of~the protocol size (i.e., the number of~leaves), rather than depth,
since this gives a~more general lower bound. In~particular, it~immediately
implies a~$(3-o(1))\log n$ depth lower bound.

\begin{restatable}{corollary}{maintheorem} \label{theorem:main}
    With probability $1 - o(1)$, for a~random function $f \colon \{0,1\}^{\log m} \to \{0,1\}$, any protocol solving $\KW_{\XOR_{m}} \circledast \KW_{f}$ has at~least $n^{3 - o(1)}$ leaves, where $n = m \log m$.
\end{restatable}

In turn, this result follows from the following general lower bound, given in terms of $\L_{\frac{3}{4}} $ that stands for the smallest size of~a~formula that agrees with~$f$ on~a~$\frac{3}{4}$ fraction of~inputs.
\begin{restatable}{theorem}{mainclasses} \label{theorem:main_classes}
    For any $0.49$-balanced function $f \colon \{0,1\}^{\log m} \to \{0,1\}$, any protocol solving $\KW_{\XOR_{m}} \circledast \KW_{f}$ has at~least $n^{2 - o(1)} \cdot \L_{\frac{3}{4}}(f)$ leaves, where $n = m \log m$.
\end{restatable}

In~contrast to~many results mentioned above and similarly to~the bound
by~de~Rezende at al.~\cite{DBLP:journals/cc/RezendeMNPR24}, our result works for a~wide range
of~inner functions $f$, what brings us closer to resolving KRW, which makes a claim about the complexity of composing any pair of functions. Also, many of~the previous techniques
work well in~the regime where the outer function is~hard
and give no~strong lower bounds when the outer function is~easy (as~it~is the case with the $\XOR{}$ function).
For example, random restrictions (as~one of~the most successful methods
for proving lower bounds) does not seem to~give meaningful lower bounds for
$\KW_{\XOR_{m}} \circledast \KW_{f}$,
as~under a~random restriction this composition turns into a~\XOR{}
of a~small number of~variables which is~easy to~compute.
The lower bound by~Meir (see \Cref{thm:meir}) also gives strong lower bounds
in~the regime where the outer function is~hard (and only gives a~trivial lower bound of~the form $o(\log n)$ for the function that we~study).

To~prove the lower bound, we~exploit formal complexity measures.
As~in~\cite{DBLP:journals/cc/EdmondsIRS01, DBLP:conf/coco/MihajlinS21},
we~consider two stages of a~protocol solving $\KW_{\XOR_{m}} \circledast \KW_{f}$.
During the first stage, we~track the progress using the classical measure by~Khrapchenko~\cite{K71} and ensure that even after many steps of~the
protocol, there are still many instances of~$f$ that need to~be solved.
At~the second stage, we~switch to~another formal complexity measure
and show that the remaining communication problem~is, roughly, not easier
than $\KW_{\OR} \circledast \KW_{f}$.
We~believe that this proof technique is~interesting on~its own,
since it~is not only easy to~show that Khrapchenko's measure cannot
give superquadratic size lower bounds, but it~is also known
that natural generalizations of~this measure are also unable
to~give stronger than quadratic lower bounds~\cite{DBLP:journals/tcs/HrubesJKP10}.
For more details on Khrapchenko's measure and its limitations, see~\Cref{section:graphs,section:formalcomplexity}.

\section{Notation, Known Facts, and Technical Lemmas}
Throughout the paper, $\log$ denotes the binary logarithm
whereas $\ln$ denotes the natural logarithm.
By~$n$ we~usually denote the size of~the input. All asymptotic estimates are given
under an~implicit assumption that $n$~goes to~infinity. By~$[n]$, we denote the set
$\{1, 2, \dotsc, n\}$.
By $\R_+$ we denote the set $\{x \in \R  \colon x > 0 \}$.
We utilize the following asymptotic estimates for binomial coefficients. For any constant $0 < \alpha < 1$,
\begin{equation}
	\label{eq:binomialestimate}
	\Omega(n^{-1/2})2^{h(\alpha)n} \leq \binom{n}{\alpha n} \leq 2^{h(\alpha)n},
\end{equation}
where
\(h(x) = -x \log x - (1-x) \log (1-x)\)
denotes the binary entropy function.

For a~string $x \in \{0,1\}^n$, its $i$-th
bit of~$x$ is~denoted by~$x_i$.
For a~matrix $X \in \{0,1\}^{m \times n}$,
by~$X_i$ we~denote the $i$-th row of~$X$ and by~$X_{i,j}$ we~denote the bit of~$X$
in~the intersection of~the $i$-th row and the $j$-th column.

\subsection{Graphs} \label{section:graphs}
    For a~rooted tree, the \emph{depth} of~its node is~the number of~edges on~the path from
    the node to~the root; the depth of~the tree is~the maximum depth of~its nodes.

	Let $G(V,E)$ be a~graph and $\varnothing \neq A \subseteq V$ be~its nonempty subset of~nodes. By~$G[A]$, we~denote a~subgraph of~$G$ induced by~$A$.
	By~$\avgdeg(G, A)$, we~denote the average degree of~$A$:
	\begin{align}
		\avgdeg(G, A) = \frac{1}{|A|} \sum_{v \in A} \deg(v)\,. \label{df:avgdeg}
	\end{align}

	For a~biparite graph $G(A \sqcup B, E)$ with nonempty parts, let
	\begin{equation}
		\label{eq:graphmeasure}
		\psi(G) = \avgdeg(G, A) \cdot \avgdeg(G, B) \, .
	\end{equation}
	Clearly, $\psi(G) \le |A| \cdot |B|$.
	The lemma below shows that this graph measure is~subadditive.

	\begin{lemma}
		\label{lemma:khrapchekno_subadditive}
		Let $G(A \sqcup B, E)$ be~a~bipartite graph and $A=A_L \sqcup A_R$
		be a~partition of~$A$ into two parts. Let
		$G_L=G[A_L \sqcup B]$
		and
		$G_R=G[A_R \sqcup B]$.
		Then,
		\[
		\psi(G) \leq \psi(G_L) + \psi(G_R).
		\]
	\end{lemma}

	\begin{proof}
		Let $E_L$ and $E_R$ be~the set of~edges of~$G_L$ and $G_R$, respectively.
		Clearly, $E = E_L \sqcup E_R$. Then,
		\begin{align*}
			\psi(G) \leq \psi(G_L) + \psi(G_R) & \iff \frac{|E|^2}{(|A_L| + |A_R|)|B|} \leq \frac{|E_L|^2}{|A_L||B|} + \frac{|E_R|^2}{|A_R||B|}  \\
			& \iff \frac{|E_L|^2 + |E_R|^2 + 2 |E_L||E_R|}{|A_L| + |A_R|} \leq \frac{|E_L|^2}{|A_L|} + \frac{|E_R|^2}{|A_R|} \\
			& \iff 2 |E_L| |E_R| |A_L| |A_R| \leq |E_R|^2 |A_L|^2 + |E_L|^2 |A_R|^2 \\
			& \iff 0 \leq (|E_R||A_L| - |E_L||A_R|)^2.
		\end{align*}
	\end{proof}

	The next lemma shows that if~$G$~contains a~node of~small enough degree,
	then deleting~it not only does not drop~$\psi$, but also does not drop too
	much the average degree of~the parts.

	\begin{lemma}
		\label{lemma:minimal_degree_balancing}
		Let a~node $a \in A$ of a~bipartite graph $G(A \sqcup B, E)$
		satisfy \(\deg(G, a) \leq {\avgdeg(G, A)}/{2}\)
		and let $A'=A \setminus \{a\}$ and
		$G'(A' \sqcup B, E')=G[A \setminus \{a\} \sqcup B]$. Then,
		\begin{align}
			\psi(G') &\geq \psi(G), \label{eq:minimal_degree_balancing:mu}\\
			\avgdeg(G', A') &\geq \avgdeg(G, A), \label{eq:minimal_degree_balancing:d_1}
		\end{align}
	\end{lemma}

	\begin{proof}
		The inequality $\avgdeg(G', A') \geq \avgdeg(G, A)$ holds since
		$A'$~results from~$A$ by~removing a~node of~degree less than the average degree.

		To~prove the inequality~\eqref{eq:minimal_degree_balancing:mu}, let
		$d=\deg(G, a)$. Then,
		$|E'|=|E|-d$ and
		\begin{align*}
			\psi(G') \geq \psi(G)
			&\iff \frac{(|E| - d)^2}{(|A| - 1)|B|} \geq \frac{|E|^2}{|A||B|}  \\
			&\iff \frac{|E|^2 - 2|E|d + d^2}{(|A| - 1)|B|} \geq \frac{|E|^2}{|A||B|} \\
			&\Longleftarrow \frac{|E| - 2d}{|A| - 1} \geq \frac{|E|}{|A|}\\
			&\iff |E||A| - 2d|A| \geq |E|(|A| - 1)\\
			&\iff d \leq \frac{|E|}{2|A|} = \frac{\avgdeg(G, A)}{2}.
		\end{align*}

	\end{proof}

\subsection{Boolean Functions}
By~$\mathbb B_n$, we~denote the set of~all Boolean functions on~$n$~variables.
For two disjoint sets $A, B \subseteq \{0,1\}^n$, the set $A \times B$
is~called a~\emph{combinatorial rectangle}, and
it~is called \emph{full} if $A$~and~$B$ form a~partition of~$\{0,1\}^n$.
Clearly, there~is a~bijection between $\mathbb B_n$ and full combinatorial rectangles.
For $f \in \mathbb B_n$, by $R_f=f^{-1}(1) \times f^{-1}(0)$,
we~denote the corresponding full rectangle.
We~say that a~Boolean function~$f$ is~\emph{balanced} if $|f^{-1}(0)|=|f^{-1}(1)|$.

In~this paper, it~will prove convenient to~apply a~function $g \in \mathbb B_m$
not only to Boolean vectors $x \in \{0,1\}^m$, but also to~matrices $X \in \{0,1\}^{n \times m}$:
\[g(X)=(g(X_1), \dotsc, g(X_n)),\]
i.e., $g(X) \in \{0,1\}^n$ results by~applying $g$~to every row of~$X$.
This allows to~define a~composition in a~natural way.
For
$f \in \mathbb B_m$ and $g \in \mathbb B_n$,
their \emph{composition}
$f \diamond g \colon \{0,1\}^{m\times n} \to \{0,1\}$ treats the input
as an~$m \times n$ matrix and first applies~$g$
to~all its rows and then applies~$f$ to~the resulting column-vector:
\[f \diamond g(X) = f(g(X)) = f(g(X_1), \dotsc, g(X_m)).\]
For a~set of~matrices $\mathcal X \subseteq \{0,1\}^{m \times n}$,
by~$i$-th projection $\proj_i \mathcal X$, we~denote the set of~all $i$-th rows among the matrices of~$\mathcal X$:
\begin{equation}
    \label{df:projection}
	\proj_i \mathcal X = \{X_i \colon X \in \mathcal{X}\}=\{t \in \{0,1\}^{n} \colon \exists X \in \mathcal X \colon t = X_i\}.
\end{equation}

In~the proof of~the main result, we~will be~dealing with Boolean matrices of~dimension $n \times \log n$. Let $\mathcal X \subseteq \{0,1\}^{n \times \log n}$ be a~set of~such matrices.
We~say that $\mathcal X$ is \emph{$\alpha$-bounded} if $|\proj_i \mathcal X| \leq \alpha n$, for all $i \in [n]$. The $i$-th projection of~$\mathcal X$ is~called \emph{sparse} if
$|\proj_i \mathcal X|<\frac{3}{8} n$, and \emph{dense} otherwise.
The following lemma shows that if~$|\mathcal X|$
is~large and $\mathcal X$ is $\alpha$-bounded, then the number of~sparse projections of~$\mathcal X$ is~low.
Later~on, we~will be~applying this lemma for~$\mathcal X$ which is almost $0.5$-bounded and whose size gradually decreases to~argue that the number of~sparse projections cannot grow
too fast.

\begin{lemma}
    \label{lemma:sparse_projections_bound}
    Let $k \in \mathbb N$ and $\alpha \in \left(\frac{3}{8}, \frac{1}{2}\right]$.
    If $\mathcal X \subseteq \{0,1\}^{n \times \log n}$ is~$\alpha$-bounded and
    $|\mathcal X| \geq \alpha^{n} \frac{n^{n}}{2^{k}}$, then
    the number of~sparse projections of~$\mathcal X$ does not exceed~$k \log^{-1} \frac{8 \alpha}{3}$.
\end{lemma}

\begin{proof}
    Let $\beta_i \in [0, 1]$ be such that $|\proj_i\mathcal X| = \beta_i \alpha n$. The $i$-th projection is sparse if and only if $\beta_i < \frac{3}{8 \alpha}$. Let $t$~be the number of~sparse projections and assume,
	without loss of~generality, that the first $t$~projections are sparse.
	Then,
    \begin{align*}
		\alpha^{n} \frac{n^{n}}{2^{k}} &\leq |\mathcal X| \leq \prod_{i = 1}^{n} |\proj_i\mathcal X| = (\alpha n)^{n} \prod_{i = 1}^{n} \beta_i \iff \\
		\frac{1}{2^{k}} &\leq \prod_{i = 1}^{n} \beta_i = \prod_{i = 1}^{t} \beta_i \cdot \prod_{i = t + 1}^{n} \beta_i < \left(\frac{3}{8 \alpha}\right)^t \implies k \log^{-1} \frac{8 \alpha}{3} \geq t.
	\end{align*}
\end{proof}

\subsection{Boolean Formulas}
The computational model studied in~this paper is~\emph{de Morgan formulas}:
it~is a~binary tree whose leaves are labeled by~variables $x_1, \dotsc, x_n$
and their negations whereas internal nodes (called gates) are labeled by~$\lor$ and~$\land$ (binary disjunction and conjunction, respectively).
Such a~formula \emph{computes} a~Boolean function $f(x_1, \dotsc, x_n) \in \mathbb B_n$. We~also say that a~formula~$F$ \emph{separates}
a~rectangle $A \times B$, if $f(a)=1$ and $f(b)=0$, for all $(a,b) \in A \times B$.
This way, if~a~formula~$F$ computes a~function~$f$, then
it~separates~$R_f$.

For a~formula~$F$, the \emph{size}~$\L(F)$ is~defined as~the number of~leaves in~$F$.
This extends to~Boolean functions:
for $f \in \mathbb B_n$, by~$\L(f)$ we~denote the smallest size of a~formula computing~$f$.
Similarly, the \emph{depth}~$\D(F)$ is~the depth of~the tree whereas $\D(f)$~is the smallest
depth of a~formula computing~$f$.

By~$\L_{\frac{3}{4}}(f)$, we denote the smallest size of~a formula~$F$ that agrees with~$f$ on~a~$\frac{3}{4}$ fraction of~inputs, i.e.,
\[
\Pr_{x \in \{ 0, 1 \}^{n}}[F(x) = f(x)] \ge \frac{3}{4}.
\]
We say that $F$ \emph{approximates} $f$.

It~is known that formulas can be~balanced: $\D(f)=\Theta(\log \L(f))$ (see references in~\cite[Section~6.1]{DBLP:books/daglib/0028687}): this~is
proved by~showing that, for any formula~$F$, there exists
an~equivalent formula~$F'$ with $\L(F') \le \L(F)^{O(1)}$
and $\D(F') \le O(\log \L(F))$.
The following theorem further refines this:
by~allowing a~larger constant in~the depth upper bound,
one can control the size of~the resulting balanced formula.

\begin{theorem}[\cite{DBLP:journals/ipl/BonetB94}]
	\label{theorem:formula_balancing}
	For any $k \ge 2$ and any formula~$F$, there exists an~equivalent
	formula~$F'$ satisfying
    $\D(F') \leq 3 \ln 2 \cdot k \cdot \log \L(F)$ and $\L(F') \leq \L(F)^{\gamma}$,
    where $\gamma = 1 + \frac{1}{1 + \log(k - 1)}$.
\end{theorem}

Using a~counting argument, one can show that, with probability $1-o(1)$,
for a~random Boolean function $f \in \mathbb B_n$,
$\L(f)=\Omega(2^n / \log n)$. To~prove this, one compares the number
of~small size formulas with the number of~Boolean functions $|\mathbb B_n|=2^{2^n}$,
using the following estimate (see~\cite[Lemma~1.23]{DBLP:books/daglib/0028687}).
It ensures that the number of~formulas of~size at~most
$\frac{2^n}{100\log n}$ is~$o(|\mathbb B_n|)$:
\[(17n)^\frac{2^n}{100\log n}=2^{\frac{\log (17n)}{100\log n}2^n}\,.\]

\begin{lemma}
	For all large enough~$l$, the number of~Boolean formulas over
	$n$~variables with at~most~$l$ leaves is~at~most
	\begin{equation}
	    \label{eq:totalformulas}
	    (17n)^l.
	\end{equation}
\end{lemma}

\begin{proof}
    The number of~binary trees with $l$~leaves is at~most $4^l$. For each such tree, there are at most $(4n)^l$ ways to convert it into
    a~de~Morgan formula: there are $2n$~input literals for the leaves and two operations for each internal gate. Consequently, the total number
    of~formulas with at most $l$~leaves is at most
    \[l \cdot 4^l \cdot (4n)^l = l \cdot 16^l \cdot n^l \le (17n)^l,\]
    which is true for $l \geq 71$.
\end{proof}

We~say that $f \in \mathbb B_n$ is \emph{$\alpha$-balanced} if
\[
	\alpha \cdot 2^n \le |f^{-1}(0)|, |f^{-1}(1)| \le (1 - \alpha) \cdot 2^n
\]
i.e., $\left||f^{-1}(0)| - |f^{-1}(1)|\right| \leq (1 - 2 \alpha) 2^{n}$.

\begin{lemma}
    \label{lemma:balanced_hard}
    For all sufficiently large~$n$ and any constant $\frac{3}{8} < \alpha < \frac{1}{2}$,
    a~random function $f \in \mathbb B_n$ is~$\alpha$-balanced and $\L_{\frac{3}{4}}(f) = \Omega(\frac{2^{n}}{\log n})$, with probability $1 - o(1)$.
\end{lemma}
\begin{proof}
	For a~formula over $n$~variables, the number of~Boolean functions
	it~approximates is~at~most (by~the estimate~\eqref{eq:binomialestimate})
	\[\sum_{d = 3 \cdot 2^{n} / 4}^{2^{n}} \binom{2^{n}}{d} =  \sum_{d=0}^{2^n/4}\binom{2^n}{d} \le 2^n \binom{2^n}{2^n/4} \leq 2^n \cdot 2^{h(1/4)2^n} \, .\]
	Combining this with~\eqref{eq:totalformulas}, we~get that the number
	of~functions approximated by~formulas of~size $\beta \frac{2^n}{\log n}$
	is~at most
	\[(17n)^{\beta \frac{2^n}{\log n}}2^n 2^{h(1/4)2^n}=2^{2^n\left(\beta\frac{\log (17n)}{\log n}+h(1/4)\right)+n}.\]
    For any constant $0<\beta < 1-h(1/4)$,
    this is a~$o(1)$ fraction of~$\mathbb B_n$.

    Now, the probability that a~random~$f \in \mathbb B_n$ is~not $\alpha$-balanced (i.e.,
    $||f^{-1}(0)| - |f^{-1}(1)|| > (1 - 2 \alpha) \cdot 2^{n}$) is~at~most
	\[
		\frac{1}{2^{2^{n}}} \cdot 2 \cdot \sum_{i = 0}^{\alpha \cdot 2^{n} - 1} \binom{2^{n}}{i}  \leq \frac{1}{2^{2^{n}}} \cdot 2 \cdot \alpha \cdot 2^{n} \cdot \binom{2^{n}}{\alpha \cdot 2^{n}}  \leq 2^{2^{n}(h(\alpha) - 1) + n + 1} = o(1).
	\]
	Thus, with probability $1 - o(1)$, a~random $f \in \mathbb B_n$ is $\alpha$-balanced and hard to~approximate.
\end{proof}

\subsection{Karchmer--Wigderson Games}
Karchmer and Wigderson~\cite{DBLP:journals/siamdm/KarchmerW90} came~up with the following characterization of~Boolean formulas.
For a~Boolean function $f \in \mathbb B_n$, the \emph{Karchmer--Wigderson game} $\KW_f$
is~the following communication problem. Alice is~given $a \in f^{-1}(1)$, whereas
Bob is~given $b \in f^{-1}(0)$, and their goal is~to find an~index $i \in [n]$
such that $a_i \neq b_i$. A~\emph{communication protocol for $\KW_f$} is~a~rooted binary tree whose leaves are labeled with indices from $[n]$ and each internal node~$v$ is~labeled either
by~a~function $A_v \colon f^{-1}(1) \to \{0,1\}$ or~by~a~function $B_v \colon f^{-1}(0) \to \{0,1\}$.
For any pair $(a,b) \in f^{-1}(1) \times f^{-1}(0)$, one can reach a~leaf of~the protocol by~traversing a~path from the root to a~leaf to~determine to~which of~the two children to~proceed from a~node~$v$, one computes either $A_v(a)$ or $B_v(b)$. We~say that a~protocol \emph{solves} $\KW_f$, if~for any $(a,b) \in f^{-1}(1) \times f^{-1}(0)$, one reaches a~leaf $i \in [n]$
such that $a_i \neq b_i$. Similarly to~formulas, we~say that a~protocol separates
a~combinatorial rectangle $A \times B$, if~it works correctly for all pairs $(a,b) \in A \times B$.

Karchmer and Wigderson showed that formulas computing~$f$ and protocols solving $\KW_f$ can be~transformed (even mechanically) into one another. In~particular, the smallest number of~leaves in~the protocol solving $\KW_f$ is~equal to~$\L(f)$, whereas the smallest
depth of~a~protocol (also known as~the communication complexity of~$\KW_f$, denoted by~$\CC(\KW_f)$) is~nothing else but~$\D(f)$.
By~$\L(A \times B)$ for~a~combinatorial rectangle $A \times B$, we denote the minimum number of leaves in~a~protocol separating $A$ and~$B$.

With each node of a~protocol solving $\KW_f$, one can associate a~combinatorial rectangle
in~a~natural way. The root of~the protocol
corresponds to~$R_f$. For the two children of Alice's node~$v$ with a~rectangle $A \times B$, one associates two rectangles $A_0 \times B$ and $A_1 \times B$,
where $A_i=\{a \in A \colon A_v(a)=i\}$. This way, Alice splits the current rectangle horizontally. Similarly, when Bob speaks, he~splits the current rectangle vertically.
Each leaf of a~protocol solving
$\KW_f$ is~associated with a~\emph{monochromatic} rectangle, i.e., a~rectangle $A \times B$ such that there exists $i \in [n]$ for which $a_i \neq b_i$ for all $(a,b) \in A \times B$.

For functions
$f \in \mathbb B_m$ and $g \in \mathbb B_n$, the
\emph{strong composition} of~$\KW_f$ and~$\KW_g$, denoted as $\KW_f \circledast \KW_g$, is the following communication problem: Alice and Bob receive inputs $X \in (f \diamond g)^{-1}(1)$ and $Y \in (f \diamond g)^{-1}(0)$, respectively, and need
to~find indices $(i, j)$ such that $X_{i, j} \neq Y_{i, j}$ and $g(X_i) \neq g(Y_i)$.
We say that a~protocol \emph{strongly separates} sets $\mathcal{X} \subseteq (f \diamond g)^{-1}(1)$ and $\mathcal{Y} \subseteq (f \diamond g)^{-1}(0)$, if it~solves the strong composition $\KW_f \circledast \KW_g$ on~inputs $\mathcal{X} \times \mathcal{Y}$.

\subsection{Formal Complexity Measures} \label{section:formalcomplexity}
For $f \in \mathbb B_n$, define a~bipartite graph
$G_f(f^{-1}(1) \sqcup f^{-1}(0), E_f)$ as~follows:
\[E_f=\{\{u, v\} \colon u \in f^{-1}(1), v \in f^{-1}(0), d_H(u, v)=1\},\]
where $d_H$~is the Hamming distance. Khrapchenko~\cite{K71}
proved that, for any $f \in \mathbb B_n$, $\psi(G_f) \le \L(f)$ (recall~\eqref{eq:graphmeasure} for the definition of~$\psi(G)$).
This immediately gives a~lower bound $\L(\XOR_n) \ge n^2$.
Note the two useful properties of~$\psi(G_f)$: on~the one hand,
it~is a~lower bound to~$\L(f)$, on~the other hand,
it~is much easier to~estimate than~$\L(f)$.

Paterson~\cite[Section~8.8]{DBLP:books/teu/Wegener87} noted that Khrapchenko's approach can be~cast as~follows.
A~function
$\mu \colon \mathbb B_n \to \R_{+}$ is called a~\emph{formal complexity measure}
if~it satisfies the following two properties:
\begin{enumerate}
	\item normalization: $\mu(x_i), \mu(\overline{x_i}) \le 1$, for all $i \in [n]$,
	\item subadditivity: $\mu(f \lor g) \leq \mu(f) + \mu(g)$ and $\mu(f \land g) \le \mu(f) + \mu(g)$, for all $f, g \in \mathbb B_n$.
\end{enumerate}
Note that Khrapchenko's measure can be~defined in~this notation as $\phi(f)=\psi(G_f)$.
Its subadditivity is~shown in~\Cref{lemma:khrapchekno_subadditive}, whereas the normalization
property can be~easily seen.

It~is not difficult to~see that $\L$~itself is a~formal complexity measure.
Moreover, it~turns out that it~is the largest formal complexity measure.
\begin{lemma}[Lemma 8.1 in~\cite{DBLP:books/teu/Wegener87}]
	For any formal complexity measure $\mu \colon \mathbb B_n \to \R$ and any $f \in B_n$, $\mu(f) \le \L(f)$.
\end{lemma}

\section{Proof of~the Main Result}

In~this section, we~prove the main result of~the paper.

\mainclasses*

\maintheorem*
\begin{proof}
    \Cref{lemma:balanced_hard} guarantees that for a random function $f  \colon \{ 0, 1 \}^{\log m} \to \{ 0, 1 \}$, $f$ is $0.49$-balanced and $\L_{\frac{3}{4}}(f) = \Omega(m / \log \log m)$ with probability $1 - o(1)$.
    Plugging this into \Cref{theorem:main_classes} gives the required lower bound.
\end{proof}

\subsection{Proof Overview}
We~start by~proving a~lower bound on~the size of~any protocol
solving $\KW_{\XOR_{m}} \circledast \KW_{f}$ and having a~logarithmic depth.
Then, using balancing techniques (see \Cref{theorem:formula_balancing}),
we~generalize the size lower bound to~all protocols.

Fix a~set $\mathcal Z \subseteq \{0,1\}^{\log m}$
such that $|\mathcal Z| = 0.98 m$ and $f$~is balanced on~$\mathcal Z$:
$|\mathcal X_0|=|\mathcal Y_0|=0.49m$, where
$\mathcal X_0=f^{-1}(1) \cap \mathcal Z$
and
$\mathcal Y_0=f^{-1}(0) \cap \mathcal Z$.

We~prove a~lower bound for any protocol that strongly separates $\KW_{\XOR_m} \circledast \KW_f$ on inputs $\mathcal X_T \times \mathcal Y_T$ (which are defined later). To~this end, we~associate, with nodes of~the protocol, a~graph similar to $G_{\XOR_m}$ and use Khrapchenko's measure to~track the progress of~the protocol.
A~node of the graph is associated with all inputs~$X$ having the same vector $f(X)$.
The reasoning is that, in a~natural scenario, the protocol will first solve $\XOR_m$, followed by solving~$f$, implying that the protocol does not need to distinguish between $X$ and $X'$ in the initial rounds, if $f(X)=f(X')$. We~connect two graph nodes by~an~edge
if~their vectors differ in~exactly one coordinate.

We aim to ensure that each edge in the graph has a large projection: for any two nodes connected by an edge, the elements of~blocks associated with them
cover a substantial number of inputs for the function~$f$.
There will be no small protocol capable of solving the problem within these two blocks since $f$~is hard to~approximate.
This is the rationale behind ensuring that all edges in the graph have large projections on both sides.
To achieve this, we enforce that each block that is associated with a~node shrinks by at most a~factor of two at each step of the protocol.
This process ensures that a significant number of edges in the graph will maintain large projections on both sides.

Once the Khrapchenko measure becomes sufficiently small, we can assert that $\XOR_m$ is nearly solved, and the protocol, in a sense, must now solve an instance of $\OR_d \circledast f$. Using the fact that solving each edge independently is hard, we conclude that solving an $\OR_d \circledast f$ over these edges should be as difficult as approximately $d \cdot \L_{\frac{3}{4}}(f)$.

\subsection{Proof}
Throughout this section, we~assume that $m$~is~large enough
and $f \in \mathbb B_{\log m}$ is a~fixed function that is~$0.49$-balanced.
Fix sets $\mathcal X_0 \subseteq f^{-1}(1), \mathcal Y_0 \subseteq f^{-1}(0)$ of~size $0.49 m$ and let
\begin{align}
	\mathcal X_T &= \{X \in \{0,1\}^{m \times \log m} \colon (\XOR_m \diamond f)(X) = 1 \wedge X_i \in \mathcal X_0 \sqcup \mathcal Y_0, \; \forall i \in [m]\}, \label{eq:x_0} \\
	\mathcal Y_T &= \{Y \in \{0,1\}^{m \times \log m} \colon (\XOR_m \diamond f)(Y) = 0 \wedge Y_i \in \mathcal X_0 \sqcup \mathcal Y_0, \; \forall i \in [m]\}. \label{eq:y_0}
\end{align}

Let $\alpha>0$ be a~constant and $P$~be a~protocol that strongly separates  $\mathcal X_T \times \mathcal Y_T$ and has depth at~most $\alpha \log m$.
Recall that each node~$S$ of~$P$ is~associated with a~rectangle $\mathcal X_S \times \mathcal Y_S$.
We~build a~subtree~$D$ of~$P$ having the same root and associate
a~graph $G_N$ to every node~$N$ of~$D$. The graphs~$G_N$
are built inductively from the graphs associated with the parents of~$N$
as~explained below,
but all these graphs are subsets of~the $m$-dimensional hypercube:
the set of~nodes of~each such graph is a~subset of $\{0,1\}^m$
and for each edge $\{u,v\}$ it~holds that $d_H(u,v)=1$.

For the root~$T$ of~the protocol~$P$,
the graph $G_T$ is~simply $G_{\XOR_m}$
(which~is nothing else but the $m$-dimensional hypercube):
its set of~nodes is $\{0,1\}^m$, two nodes are joined by~an~edge with
label~$i$ if~they differ in~the $i$-th coordinate.

For any node~$v$ of~the graph $G_S$, we~associate the following set of inputs called \emph{block}:
\[\mathcal{B}_S(v) = \{X \in \{0,1\}^{m \times \log m} \colon X \in \mathcal X_S \sqcup \mathcal Y_S \text{ and } f(X)=v\}.\]
We~say that an~edge $\{u,v\}$ with label~$i$ of~$G_S$ is~\emph{heavy}
if~the projection of both $\mathcal{B}_S(u)$ and $\mathcal{B}_S(v)$ onto the $i$th coordinate is dense, i.e.,
\[
|\proj_i\mathcal{B}_S(u)|, |\proj_i\mathcal{B}_S(v)| \ge \frac 3 8 m,
\] and \emph{light} otherwise.

Since the nodes of the graph $G_{S}$ form a subset of $\{0,1\}^{m}$, we can naturally divide them into two parts, as their blocks correspond to subsets of either $\mathcal X_S$ or $\mathcal Y_S$.
\begin{align*}
	A_{S} &= \{ v \in V(G_{S}) \mid \XOR_m(v) = 1\}\\
	B_{S} &= \{ v \in V(G_{S}) \mid \XOR_m(v) = 0\}
\end{align*}
For a graph $G_{S}$, we define $d_A(G_{S})$ as the average degree of the part $A_{S}$ and $d_B(G_{S})$ as the average degree of the part $B_{S}$.
We say that a graph $G_S$ is \emph{special} if
 \[
	 \min \{ d_A(G_S), d_B(G_S) \} \leq 12 \alpha \log^2 m.
\]
We will construct the tree $D$ inductively.
For a~node~$S$ in the tree~$D$, we either stop the process if $G_S$ is special, or construct the two children of~$S$ from the protocol $P$ and their graphs.
We continue building~$D$ on these two children inductively.
Hence, all graphs corresponding to internal nodes of the tree~$D$ are not special, while all graphs associated with leaves of $D$ are special.

\begin{definition}
	A graph $G_S$, associated with a node $S$ in the tree $D$, is \emph{adjusted} if all its edges are heavy and
\begin{equation}
	\deg(v) > \frac{d_A(G_S)}{2}, \; \forall v \in A_S \qquad \text{and} \qquad \deg(v) > \frac{d_B(G_S)}{2}, \; \forall v \in B_S. \label{eq:big_degree}
\end{equation}
\end{definition}
We will ensure that all graphs $G_S$ for any node $S$ in the tree $D$ are adjusted.

\begin{lemma}
    \label{lemma:block_graph_properties}
    For each node~$v$ of~the graph~$G_T$ (associated with the root~$T$ of~the protocol~$P$),
	\begin{enumerate}
        \item the degree of~$v$ is~$m$;
		\item $|\proj_i\mathcal{B}_T(v)| = 0.49 m$, for all $i \in [m]$;
		\item $|\mathcal{B}_T(v)|=\frac{(0.98 m)^{m}}{2^{m}}$.
	\end{enumerate}
\end{lemma}
\begin{proof}
    Nodes of $G_T$ are $m$-dimensional binary vectors, hence $\deg(v)=m$.

    To~prove the second property,
    recall that $f$~is balanced on~$\mathcal X_0 \sqcup \mathcal Y_0$.
    If $v_i=1$ (or $v_i=0$), for some $i \in [m]$,
    the $i$-th projection
    can take any value from $\mathcal X_0$ ($\mathcal Y_0$, respectively).
    Hence, $|\proj_i\mathcal{B}_T(v)| = 0.49 m$.

    Finally, to~prove the third property,
    note the $G_T$~has $2^m$~nodes
    and for each vertex $v$ the size of the block $\mathcal B_T(v)$ is~at~most $(0.49m)^m$.
    Therefore, since each input from $\mathcal X_T \sqcup \mathcal Y_T$ belongs to~exactly one block that is associated with a~node from~$G_T$ and $|\mathcal X_T \sqcup \mathcal Y_T| =  |\mathcal X_0 \sqcup \mathcal Y_0|^{m} = (0.98 m)^{m}$,
    $|\mathcal{B}_T(v)|=\frac{(0.98 m)^m}{2^m}$.
\end{proof}

\Cref{lemma:block_graph_properties} ensures that the graph $G_{T}$ is adjusted and not special, thus the~root $T$ has two children.
Using the function $\mathcal B$,
we~show how to construct an intermediate graph~$H_N$ for some child of a node $S$ in the tree  $D$ and then we apply some cleanup procedures for the graph $H_N$ to construct a graph $G_N$.
Recall that each step of~$P$ partitions the set of~either Alice's or~Bob's inputs into two parts.
Let $G_{S}$ be a graph for some node $S$ of the protocol $P$ that is associated with a rectangle  $\mathcal  X_S \times \mathcal Y_S$ and assume, without loss of generality, that it is Alice's turn.
Therefore, graph $G_S$ is not special, otherwise we will stop the building process of the subtree of $S$.
Let $S_L$ be the left child of $S$ in the protocol  $P$ and  $S_R$ be the right child.
We add the same children of the node $S$ in the tree $D$.
Then, we put $v$ from  $B_{S}$ into both $H_{S_L}$ and $H_{S_R}$ (since the block $\mathcal B_S(v)$ has not changed).
For each node $v \in A_{S}$ we decide in which of the two graphs we will put it.
The block $\mathcal B_S(v)$ is also split into two: $\mathcal{B}_{S_L}(v)$ and $\mathcal{B}_{S_R}(v)$, corresponding to the two ways of the protocol.
We assign $v$ to the left graph $H_{S_L}$ if $2 \cdot |\mathcal{B}_{S_L}(v)| \geq |\mathcal{B}_S(v)|$, and to the right graph $H_{S_R}$ if $2 \cdot |\mathcal{B}_{S_R}(v)| > |\mathcal{B}_S(v)|$.
An edge $\{u, v\}$ from the edges of $G_{S}$ goes to $H_{S_L}$ if and only if both $u$ and $v$ are assigned to $H_{S_L}$. The same rule applies for edges in $H_{S_R}$.
This approach ensures that the size of each block $\mathcal B_S(v)$ shrinks by at most a factor of two when transitioning from a parent to a child in the tree $D$.
Then, the graphs $G_{S_L}$ and $G_{S_R}$ will be built using graphs $H_{S_L}$ and $H_{S_R}$, respectively.

The idea of the structure of the graph $G_S$ arises from Khrapchenko's graph, so we will use the same measure:
\[
  \psi(G_{S}) = d_A(G_{S}) \cdot d_B(G_{S}).
\]
\Cref{lemma:khrapchekno_subadditive} states that $\psi$ is subadditive.

After obtaining the graph~$H_C$ for a~node~$C$ of the tree~$D$, we make our first cleanup
by~deleting all light edges: let $H_C'$ be a~graph resulting from $H_C$ by~removing
all its light edges. The next lemma shows that
this does not drop the measure $\psi$ too much.

\begin{lemma}
    \label{lemma:degree_changes}
	\[
		\psi(H_{C}') \geq \psi(H_{C}) \left(1 - \frac{1}{\log m}\right).
	\]
\end{lemma}

\begin{proof}
    Let $S$~be the parent of~$C$ in~$D$. Since $S$~is not a~leaf, we have that $\min\{d_A(G_S), d_B(G_S)\} > 12 \alpha \log^2 m$ and the degree of~every node in $G_S$
    is~at~least half of the average degree of~its part.
	Without loss of generality, assume that inputs were deleted from $\mathcal X_S$, and therefore $d_A(H_C) \geq \frac{d_A(G_S)}{2} > 6 \alpha \log^2 m$.

	An~edge $\{u, v\}$ can become light
    because of~only one of~its endpoints,
    because the blocks on~the other side remain unchanged.
	From \Cref{lemma:block_graph_properties}, we know that the initial size of each block is $(0.49 m)^m$, and after each step of the protocol, the size of a~block shrinks by at most a~factor of two.
	Hence, for any node~$v$, the size of its block $\mathcal B_S(v)$ is at least $\frac{(0.49 m)^{m}}{2^{\alpha \log m}}$, because the protocol depth is~bounded by $\alpha \log m$.
	Hence, we can bound the number of~light edges incident to~$v$ by $3 \alpha \log m$ using \Cref{lemma:sparse_projections_bound} (since $\log^{-1} (8 \cdot 0.49 / 3) < 3$).
	Therefore,
	\[
		d_A(H_C') \geq d_A(H_C) - 3 \alpha \log m.
	\]

	Now, consider $d_B(H_C')$.
	Let $E_C$ be the set of edges in $H_{C}$, whereas $A_C$~and~$B_C$
    be its parts of~nodes.
    Then,
	\[
		d_B(H_C') \geq \frac{E_C - |A_C| \cdot 3 \alpha \log m}{|B_C|} = d_B(H_C) - 3 \alpha \log m \frac{|A_C|}{|B_C|}.
	\]
	Hence,
	\begin{align*}
		\psi(H_C') = d_A(H_C') d_B(H_C') &\geq \left(d_A(H_C) - 3 \alpha \log m\right) \left(d_B(H_C) - 3 \alpha \log m \frac{|A_C|}{|B_C|}\right) \\
		&\geq \psi(H_C) - 3 \alpha d_B(H_C) \log m - \frac{3 \alpha |A_C| d_A(H_C) \log m}{|B_C|} \\
		&= \psi(H_C)\left(1 - \frac{3 \alpha \log m}{d_A(H_C)} - \frac{3 \alpha |A_C| \log m}{|E_C|}\right) \\
		&= \psi(H_C)\left(1 - \frac{6 \alpha \log m}{d_A(H_C)}\right) \\
		&> \psi(H_C)\left(1 - \frac{6 \alpha \log m}{6 \alpha \log^2 m}\right) = \psi(H_C)\left(1 - \frac{1}{\log m}\right).
	\end{align*}
\end{proof}

The next lemma shows how to~construct an~adjusted graph~$G_C$,
from the intermediate graph~$H_C'$.

\begin{lemma}
    \label{lemma:balancing_graph}
    There exists a~subgraph $G_C$ of~the graph~$H_C'$
	such that $G_C$ is adjusted and $\psi(G_{C}) \geq \psi(H_{C}) \left(1 - \frac{1}{\log m}\right)$.
\end{lemma}

\begin{proof}
    To~get $G_C$, we~keep removing nodes from $H_C'$ until it~satisfies~\eqref{eq:big_degree}.
	If~\eqref{eq:big_degree} is violated, there exists, without loss of~generality, a~node $v \in A_{C}$ such that $\deg(v) \leq \frac{d_A(G_C)}{2}$.
	Let $G_C' = G_C \setminus \{v\}$.
	\Cref{lemma:minimal_degree_balancing} guarantees that this does not decrease the measure.
    This process is~clearly finite.
\end{proof}

This way, we~construct the graph~$G_C$ for the node~$C$.
If~$C$ is~not special, we~continue expanding the subtree rooted at~$C$.
Recall also that, for each internal node~$S$
of~the tree~$D$, whose children are $S_L$ and $S_R$, the following holds:
\[
  \psi(G_S) \leq \psi(H_{S_L}) + \psi(H_{S_L}).
\]
Hence, combining it with \Cref{lemma:balancing_graph} we have:
\begin{equation}
    \label{eq:semi_subadditivity}
    \psi(G_S) \left(1 - \frac{1}{\log m}\right) \leq \psi(G_{S_L}) + \psi(G_{S_R}).
\end{equation}
On~the other hand, if~$S$~is~special,
we~will use the following two lemmas to~argue that
strongly separating $\mathcal{X}_S \times \mathcal{Y}_S$
is~still difficult.

\begin{lemma}
    \label{lemma:changing_measure}
	Let~$S$ be a node of the tree~$D$ such that it has a~node~$v \in G_S$
    having $d$~adjacent edges.
    Then,
    any protocol that strongly separates $\mathcal{X}_S$~and~$\mathcal{Y}_S$ has at~least $\Omega\left(d \cdot \L_{\frac{3}{4}}(f)\right)$ leaves.
\end{lemma}
\begin{proof}
    Consider the~subgraph of~$G_{S}$
    induced by~$v$ and its neighbors $u_1, \dotsc, u_d$ connected to~$v$.
	Denote by~$l_i$ the label of~the edge $\{v, u_i\}$.
    Define a~measure $\xi$ on subrectangles of $\mathcal X_S \times \mathcal Y_S$:
    \[
        \label{eq:xi}
        \xi(\mathcal X \times \mathcal Y) = \sum_{i=1}^{d} \L\left(\proj_{l_i}\mathcal{B} \times \proj_{l_i}\mathcal{B}_i\right),
    \]
	where $\mathcal X \subseteq \mathcal X_S, \, \mathcal Y \subseteq \mathcal Y_S$, $\mathcal B = \mathcal X \cap \mathcal B_S(v)$ and $\mathcal B_i = \mathcal Y \cap \mathcal B_S(u_i)$, for all $i \in [d]$.
	By~$\KW(A \times B)$, for~any~$A \cap B = \varnothing$, we~denote a~Karchmer-Wigderson communication game where Alice~gets $a \in A$, Bob~gets $b \in B$, and they~need to~find $i \colon a_i \neq b_i$.
    We prove that any protocol strongly separating $\mathcal{X}_S \times \mathcal{Y}_S$ requires at~least $\xi(\mathcal X_S \times \mathcal Y_S)$ leaves.

  It is easy to see that $\xi$ is subadditive, being a~sum of subadditive measures:
  if $\mathcal X = \mathcal {X'} \sqcup \mathcal {X''}$, then $\xi(\mathcal X \times \mathcal Y) \leq \xi(\mathcal {X'} \times \mathcal Y) + \xi(\mathcal {X''} \times \mathcal Y)$ and the same applies when we split $\mathcal Y$.
  Namely, let $\mathcal{Y} = \mathcal{Y'} \sqcup \mathcal{Y''}$, $\mathcal{B}_i = \mathcal{Y'} \cap \mathcal{B}_S(u_i)$, and $\mathcal{B}_i'' = \mathcal{Y''} \cap \mathcal{B}_S(u_i)$.
  Then,
  \begin{align*}
	  \xi(\mathcal{X} \times \mathcal{Y'} \sqcup \mathcal{Y''}) &= \sum_{i=1}^{d} \L\left(\proj_{l_i}\mathcal{B} \times \proj_{l_i} \mathcal{B}_i' \sqcup \mathcal{B}_i'' \right) \\
											&\le \sum_{i=1}^{d} \L\left(\proj_{l_i}\mathcal{B} \times \proj_{l_i}\mathcal{B}_i'\right) + \sum_{i=1}^{d} \L\left(\proj_{l_i}\mathcal{B} \times \proj_{l_i}\mathcal{B}_i''\right) \\
											&= \xi(\mathcal{X} \times \mathcal{Y}') + \xi(\mathcal{X} \times \mathcal{Y}'')
  .\end{align*}

    Consider a~protocol~$P'$ strongly separating $\mathcal X_S \times \mathcal Y_S$
	and its leaf~$L$ associated with a~rectangle of~inputs $\mathcal {X'}_L \times \mathcal  {Y'}_L$.
	We~show that $\xi(\mathcal {X'}_L \times \mathcal {Y'}_L) \le 1$.
  Since $L$~is a~leaf, there exists $i, j$ such that for each $X \in \mathcal{X'}_L$ and $Y \in \mathcal{Y'}_L$:
   \[
	   X_{i, j} \neq Y_{i, j} \quad \text{and} \quad f(X_i) \neq f(Y_i).
  \]
  Let $k$~be such that $\mathcal{B}_k \neq \varnothing$ (if all $\mathcal{B}_t$ are empty, then $\xi = 0$).
  Then, $\mathcal{B}(u_t)=\varnothing$, for all $t \neq k$, as~otherwise there would~be
  no~$i$ such that $f(X_i) \neq f(Y_i)$ for all $(X, Y) \in \mathcal{X'}_L \times \mathcal{Y'}_L$, since $u_k$ differs from $v$ in the position $l_k$, and $u_t$ differs from $v$ in the position $l_t$ and $l_k \neq l_t$.
  Thus, if $\xi(\mathcal {X'}_L \times \mathcal {Y'}_L) > 1$, then $\L(\proj_{l_k}\mathcal{B} \times \proj_{l_k}\mathcal{B}_k) > 1$, which contradicts to~the existence of a~pair $(i, j)$.

  Thus, $\xi$ is normal (has the value at~most~$1$ for any leaf of~any protocol that strongly separates $\mathcal{X}_S \times \mathcal{Y}_S$) and subadditive. Hence,
  its value for the whole protocol~$P'$ is a~lower bound on~the size of~$P'$.
  Thus, it~remains to~estimate $\xi$ for~$P'$.

  Since all $d$~edges are heavy, we have:
   \[
    |\proj_{l_i}\mathcal{B}_S(v)| + |\proj_{l_i}\mathcal{B}_S(u_i)| \geq \frac{3}{4}m, \quad \forall i \in [d].
  \]

  Hence, \[\L(\proj_{l_i}\mathcal{B}_S(v) \times \proj_{l_i}\mathcal{B}_S(u_i)) = \Omega\left(\L_{\frac{3}{4}}(f)\right),\]
  for all $i \in [d]$. Summing over all $i \in [d]$, gives the desired lower bound.

\end{proof}

\begin{lemma}
    \label{lemma:base_andreev}
    For a~special node~$S$ of~the tree~$D$,
	the number of~leaves in~any protocol strongly separating
    $\mathcal X_S \times \mathcal Y_S$ is
	\[
	    \Omega \left(\frac{\psi(G_S) \cdot \L_{\frac{3}{4}}(f)}{\log^2 m}\right).
	\]
\end{lemma}

\begin{proof}
    Assume, without loss of~generality, that
	\[
		d_A(G_S) \geq d_B(G_S) \quad \text{and} \quad d_B(G_S) \leq 12 \alpha \log^2 m.
	\]
    Applying \Cref{lemma:changing_measure}
    to a~node of~degree at~least $d_A(G_S)$,
    we~get that
    the number of leaves is at least
	\[
	    \Omega\left(d_A(G_S) \cdot \L_{\frac{3}{4}}(f)\right) = \Omega\left(\frac{\psi(G_S)}{d_B(G_S)} \cdot \L_{\frac{3}{4}}(f)\right) = \Omega \left(\frac{\psi(G_S) \cdot \L_{\frac{3}{4}}(f)}{\log^2 m}\right).
	\]
\end{proof}

At~this point, everything is~ready to~lower bound the size of~any protocol
of~logarithmic depth.

\begin{theorem}
    \label{theorem:andreev_bound_l_k}
	The size of the protocol~$P$ (strongly separating $\mathcal X_T \times \mathcal Y_T$)~is
	\[
	    \Omega\left(\frac{m^2 \cdot \L_{\frac{3}{4}}(f)}{\log^2 m} \left(1 - \frac{1}{\log m}\right)^{\alpha \log m}\right).
	\]
\end{theorem}

\begin{proof}
    \Cref{lemma:base_andreev} states that the number of leaves needed to resolve any leaf $S$ of the tree $D$ is $\Omega \left(\psi(G_S) \cdot \L_{\frac{3}{4}}(f) / \log^2 m\right)$.
	Let $\mathcal S$ be the set of all leaves of the tree $D$.
	Using estimate \eqref{eq:semi_subadditivity}, we have:
	\[
	  \psi(G_T) \cdot \left(1 - \frac{1}{\log m}\right)^{\alpha \log m} \leq \sum_{S \in \mathcal S} \psi(G_S).
  \]
    Since $\psi(G_{T})=m^2$ (by~\Cref{lemma:block_graph_properties}),
	Then, the number of leaves in $P$ is
	 \[
	     \Omega \left(\sum_{S \in \mathcal S} \frac{\psi(G_S) \cdot \L_{\frac{3}{4}}(f)}{\log^2 m}\right) \geq \Omega\left(\frac{m^2 \cdot \L_{\frac{3}{4}}(f)}{\log^2 m} \left(1 - \frac{1}{\log m}\right)^{\alpha \log m}\right).
	\]
\end{proof}

Recall that $\alpha$ is a~constant.
Assuming $m \geq 4$, we have $\log m \geq 2$, and thus $1 - \frac{1}{\log m} \geq e^{-\frac{2}{\log m}}$.
Then,
\[
    \frac{m^2 \cdot \L_{\frac{3}{4}}(f)}{\log^2 m} \left(1 - \frac{1}{\log m}\right)^{\alpha \log m} \geq \frac{m^2 \cdot \L_{\frac{3}{4}}(f)}{\log^2 m} e^{-\frac{2}{\log m} \cdot \alpha \log m} \geq m^{2 - \varepsilon} \cdot \L_{\frac{3}{4}}(f),
\]

for any constant $\varepsilon > 0$ when $m$ is sufficiently large.
Hence, the number of leaves needed for a protocol $P$ is  $m^{2 - o(1)} \cdot \L_{\frac{3}{4}}(f)$.

Finally, we get rid of~the assumption that the depth of~$P$ is~logarithmic and
prove the main result.
\begin{proof}[Proof of \Cref{theorem:main_classes}]
    Let $P$~be a~protocol with $m^{2 - \varepsilon} \cdot \L_{\frac{3}{4}}(f)$ leaves, for some $\varepsilon > 0$, solving $\KW_{\XOR_m} \circledast \KW_f$.
    We transform it into a~protocol $P'$ with $(m^{(2 - \varepsilon)} \cdot \L_{\frac{3}{4}})^{\gamma}$ leaves and depth bounded by $3(3 - \varepsilon) k \ln 2 \cdot \log m$, by applying \Cref{theorem:formula_balancing}, where $\gamma = 1 + \frac{1}{1 + \log(k - 1)}$.
    (\Cref{theorem:formula_balancing} is~stated in~terms of~formulas, but it~is not difficult to~see that it~works also for protocols
    for strong composition.)

    Since $\varepsilon > 0$ and $\lim_{k \to \infty} \gamma = 1$, there exist $k$ and $\varepsilon' > 0$ such that
    \[\left( m^{2 - \varepsilon} \cdot \L_{\frac{3}{4}}(f) \right)^{\gamma} \le m^{2 - \varepsilon'} \cdot \L_{\frac{3}{4}}(f),\] since $\L_{\frac{3}{4}}(m) \le m$.
    Hence, protocol $P'$ has logarithmic depth and at most $m^{2 - \varepsilon'} \cdot \L_{\frac{3}{4}}(f)$ leaves, which contradicts \Cref{theorem:andreev_bound_l_k}.
    Therefore, $P$ has $\Omega\left(m^{2 - o(1)} \cdot \L_{\frac{3}{4}}(f)\right) = \Omega(n^{2 - o(1)} \cdot \L_{\frac{3}{4}}(f))$ leaves.
\end{proof}

\bibliographystyle{plain}
\bibliography{references}

\begin{thebibliography}{10}

\bibitem{DBLP:journals/ipl/BonetB94}
Maria~Luisa Bonet and Samuel~R. Buss.
\newblock Size-depth tradeoffs for boolean fomulae.
\newblock {\em Inf. Process. Lett.}, 49(3):151--155, 1994.

\bibitem{DBLP:journals/cc/RezendeMNPR24}
Susanna~F. de~Rezende, Or~Meir, Jakob Nordstr{\"{o}}m, Toniann Pitassi, and
  Robert Robere.
\newblock {KRW} composition theorems via lifting.
\newblock {\em Comput. Complex.}, 33(1):4, 2024.

\bibitem{DBLP:journals/cc/DinurM18}
Irit Dinur and Or~Meir.
\newblock Toward the {KRW} composition conjecture: Cubic formula lower bounds
  via communication complexity.
\newblock {\em Comput. Complex.}, 27(3):375--462, 2018.

\bibitem{DBLP:journals/cc/EdmondsIRS01}
Jeff Edmonds, Russell Impagliazzo, Steven Rudich, and Jir{\'{\i}} Sgall.
\newblock Communication complexity towards lower bounds on circuit depth.
\newblock {\em Comput. Complex.}, 10(3):210--246, 2001.

\bibitem{DBLP:journals/siamcomp/GavinskyMWW17}
Dmitry Gavinsky, Or~Meir, Omri Weinstein, and Avi Wigderson.
\newblock Toward better formula lower bounds: The composition of a function and
  a universal relation.
\newblock {\em {SIAM} J. Comput.}, 46(1):114--131, 2017.

\bibitem{DBLP:journals/siamcomp/Hastad98}
Johan H{\aa}stad.
\newblock The shrinkage exponent of de morgan formulas is 2.
\newblock {\em {SIAM} J. Comput.}, 27(1):48--64, 1998.

\bibitem{DBLP:conf/dimacs/HastadW90}
Johan H{\aa}stad and Avi Wigderson.
\newblock Composition of the universal relation.
\newblock In Jin{-}Yi Cai, editor, {\em Advances In Computational Complexity
  Theory, Proceedings of a {DIMACS} Workshop, New Jersey, USA, December 3-7,
  1990}, volume~13 of {\em {DIMACS} Series in Discrete Mathematics and
  Theoretical Computer Science}, pages 119--134. {DIMACS/AMS}, 1990.

\bibitem{DBLP:journals/tcs/HrubesJKP10}
Pavel Hrubes, Stasys Jukna, Alexander~S. Kulikov, and Pavel Pudl{\'{a}}k.
\newblock On convex complexity measures.
\newblock {\em Theor. Comput. Sci.}, 411(16-18):1842--1854, 2010.

\bibitem{DBLP:journals/rsa/ImpagliazzoN93}
Russell Impagliazzo and Noam Nisan.
\newblock The effect of random restrictions on formula size.
\newblock {\em Random Struct. Algorithms}, 4(2):121--134, 1993.

\bibitem{DBLP:books/daglib/0028687}
Stasys Jukna.
\newblock {\em Boolean Function Complexity - Advances and Frontiers}, volume~27
  of {\em Algorithms and combinatorics}.
\newblock Springer, 2012.

\bibitem{DBLP:journals/cc/KarchmerRW95}
Mauricio Karchmer, Ran Raz, and Avi Wigderson.
\newblock Super-logarithmic depth lower bounds via the direct sum in
  communication complexity.
\newblock {\em Comput. Complex.}, 5(3/4):191--204, 1995.

\bibitem{DBLP:journals/siamdm/KarchmerW90}
Mauricio Karchmer and Avi Wigderson.
\newblock Monotone circuits for connectivity require super-logarithmic depth.
\newblock {\em {SIAM} J. Discret. Math.}, 3(2):255--265, 1990.

\bibitem{K71}
V.~M. Khrapchenko.
\newblock Method of determining lower bounds for the complexity of p-schemes.
\newblock {\em Mathematical notes of the Academy of Sciences of the USSR},
  10(1):474--479, 1971.

\bibitem{DBLP:conf/focs/Meir23}
Or~Meir.
\newblock Toward better depth lower bounds: {A} krw-like theorem for strong
  composition.
\newblock In {\em 64th {IEEE} Annual Symposium on Foundations of Computer
  Science, {FOCS} 2023, Santa Cruz, CA, USA, November 6-9, 2023}, pages
  1056--1081. {IEEE}, 2023.

\bibitem{DBLP:conf/coco/MihajlinS21}
Ivan Mihajlin and Alexander Smal.
\newblock Toward better depth lower bounds: The {XOR-KRW} conjecture.
\newblock In Valentine Kabanets, editor, {\em 36th Computational Complexity
  Conference, {CCC} 2021, July 20-23, 2021, Toronto, Ontario, Canada (Virtual
  Conference)}, volume 200 of {\em LIPIcs}, pages 38:1--38:24. Schloss Dagstuhl
  - Leibniz-Zentrum f{\"{u}}r Informatik, 2021.

\bibitem{DBLP:journals/rsa/PatersonZ93}
Mike Paterson and Uri Zwick.
\newblock Shrinkage of de morgan formulae under restriction.
\newblock {\em Random Struct. Algorithms}, 4(2):135--150, 1993.

\bibitem{Sub61}
B.~A. Subbotovskaya.
\newblock Realization of linear functions by formulas using $\vee$, $\&$, $^-$.
\newblock {\em Dokl. Akad. Nauk SSSR}, 136(3):553--555, 1961.

\bibitem{DBLP:conf/focs/Tal14}
Avishay Tal.
\newblock Shrinkage of de morgan formulae by spectral techniques.
\newblock In {\em 55th {IEEE} Annual Symposium on Foundations of Computer
  Science, {FOCS} 2014, Philadelphia, PA, USA, October 18-21, 2014}, pages
  551--560. {IEEE} Computer Society, 2014.

\bibitem{DBLP:books/teu/Wegener87}
Ingo Wegener.
\newblock {\em The complexity of Boolean functions}.
\newblock Wiley-Teubner, 1987.

\bibitem{DBLP:journals/corr/abs-2310-07422}
Hao Wu.
\newblock An improved composition theorem of a universal relation and most
  functions via effective restriction.
\newblock {\em CoRR}, abs/2310.07422, 2023.

\end{thebibliography}

\end{document}